\documentclass[conference]{IEEEtran}

\def\BibTeX{{\rm B\kern-.05em{\sc i\kern-.025em b}\kern-.08em
		T\kern-.1667em\lower.7ex\hbox{E}\kern-.125emX}}

\usepackage{microtype}
\usepackage{graphicx}
\usepackage{booktabs} 

\usepackage{cite}
\usepackage{amsmath,amsthm,amssymb,amsfonts}
\usepackage{algorithm}
\usepackage[noend]{algpseudocode}
\usepackage{graphicx}
\usepackage{textcomp}
\usepackage{xcolor,colortbl}
\usepackage{comment}
\usepackage{epstopdf}
\usepackage{float}
\usepackage{amssymb}
\usepackage{color}
\usepackage{relsize}
\usepackage{mathdots}
\usepackage{enumitem}
\usepackage{graphicx}
\usepackage{mathtools}
\usepackage{cleveref}
\usepackage{cases}
\usepackage{multicol}
\usepackage{cuted}

\theoremstyle{plain}

\newtheorem{lemma}{Lemma}
\newtheorem{definition}{Definition}

\newtheorem{claim}{Claim}

\theoremstyle{definition}

\begin{document}
	
	\title{Non-Uniform Windowed Decoding For Multi-Dimensional Spatially-Coupled LDPC Codes\vspace{-0.5cm}
	}
	
	\author{\IEEEauthorblockN{Lev Tauz,  Homa Esfahanizadeh, and Lara Dolecek}
		\IEEEauthorblockA{Department of Electrical and Computer Engineering, University of California, Los Angeles, USA\\
			levtauz@g.ucla.edu, hesfahanizadeh@ucla.edu, and dolecek@ee.ucla.edu 
			\vspace{-0.50cm}
		}
	}
	
	\maketitle

	\begin{abstract}
		In this paper, we propose a non-uniform windowed decoder for multi-dimensional spatially-coupled LDPC (MD-SC-LDPC) codes over the binary erasure channel. An MD-SC-LDPC code is constructed by connecting together several SC-LDPC codes into one larger code that provides major benefits over a variety of channel models. 
		In general, SC codes allow for low-latency windowed decoding. While a standard windowed decoder can be naively applied, such an approach does not fully utilize the unique structure of MD-SC-LDPC codes. In this paper, we propose and analyze a novel non-uniform decoder to provide more flexibility between latency and reliability. Our theoretical derivations and empirical results show that our non-uniform decoder greatly improves upon the standard windowed decoder in terms of design flexibility, latency, and complexity.
	\end{abstract}
	\vspace{-0.1cm}
	\section{Introduction and Motivation}
	
	Spatially-coupled LDPC (SC-LDPC) codes are a popular choice for error-correcting codes due to their capacity-achieving performance \cite{KudekarIT2011,Kudekar2012SpatiallyCE} and low-latency windowed decoding \cite{Iyengar2012}. Multi-Dimensional SC-LDPC (MD-SC-LDPC) codes are a class of LDPC codes \cite{Esfahanizadeh2019MultiDimensionalSC,Hareedy2019MinimizingTN,OhashiISIT2013,LiuCOMML2015,SchmalenISTC2014,OlmostTCOM2017,TruhachevTIT2019} created by connecting several SC-LDPC codes. This class of codes has many significant benefits compared to conventional SC codes, including lower population of detrimental objects for belief propagation (BP) decoders \cite{Esfahanizadeh2019MultiDimensionalSC,Hareedy2019MinimizingTN}, improved reliability over parallel channels \cite{SchmalenISTC2014}, and  robustness to burst erasures \cite{OhashiISIT2013}. 
	
	One major benefit of MD-SC-LDPC codes is that many of these constructions preserve the chain structure of an SC-LDPC code which allows for windowed decoding \cite{Iyengar2012,IyengarIT2013}. One way of applying windowed decoding is for each constituent SC code to have its own window and for the windows to move in unison along the coupled constituent SC chains. This approach provides the same proportional latency benefits as it does for a single SC code, relative to the block length. Conventionally, to improve latency, one can only reduce the window size uniformly across all the constituent codes at the cost of lower reliability. However, this approach does not take into account the unique structure of MD-SC-LDPC codes and, as a result, causes unnecessary reliability loss. By allowing non-uniform window sizes across the constituent codes, we exploit the structure of MD-SC codes to provide more decoder design flexibility.

	In this paper, we define a code ensemble that captures the multi-dimensional (MD) coupling structure which can be exploited for designing a flexible decoder. We study the new MD-SC-LDPC ensemble and compare it with the standard SC-LDPC ensemble in terms of finite and asymptotic properties. Next, we propose a novel non-uniform windowed decoder that takes into account the unique structure of MD-SC codes. Then, using density evolution (DE) techniques, we analyze the reliability of our new construction and provide insight into designing a non-uniform windowed decoder.
	
	Through our new decoder construction and utilization of MD-SC properties, we demonstrate a large improvement over uniform windowed decoding. For example, we demonstrate that by increasing the decoder latency by a small amount, we can decrease the average number of iterations per window by $35\%$. Additionally, we show that our decoder achieves threshold closer to the optimal decoder threshold compared to uniform windowed decoding. As such, our decoder can reliably operate at higher channel erasure probabilities for the same decoding complexity and latency. While we demonstrate the efficacy of our decoder on our new ensemble, our non-uniform decoder is also beneficial for other MD-SC codes \cite{OhashiISIT2013,LiuCOMML2015} as it exploits the unique coupling principles of MD-SC codes.
	
	We define some necessary notations. For positive integers $A$ and $B$, we define the set $[A] \triangleq \{0,1,\dots,A-1\}$ and the operation $(n)_B = n \mod B$. For a node $v$ in a graph, $N(v)$ is the set of neighboring nodes of $v$. Additionally, given two vectors $x$ and $y$, we define $x \preceq y$ to be an element-wise inequality such that $x_i \leq y_i$ for all $i$. Also, we let $\mathbb{Z}$ be the set of all integers. We define the operator $\mathcal{O}(g(x))$ as the standard notion of complexity for function $g(x)$.
	\vspace{-0.2cm}
	\section{Preliminaries}
	\subsection{Definition: $\mathcal{C}(d_l,d_r,L_1,\gamma_1,L_2,\gamma_2,\mathcal{T})$ Ensemble } 
	
	\begin{figure}
		\centering
		\includegraphics[scale=0.45]{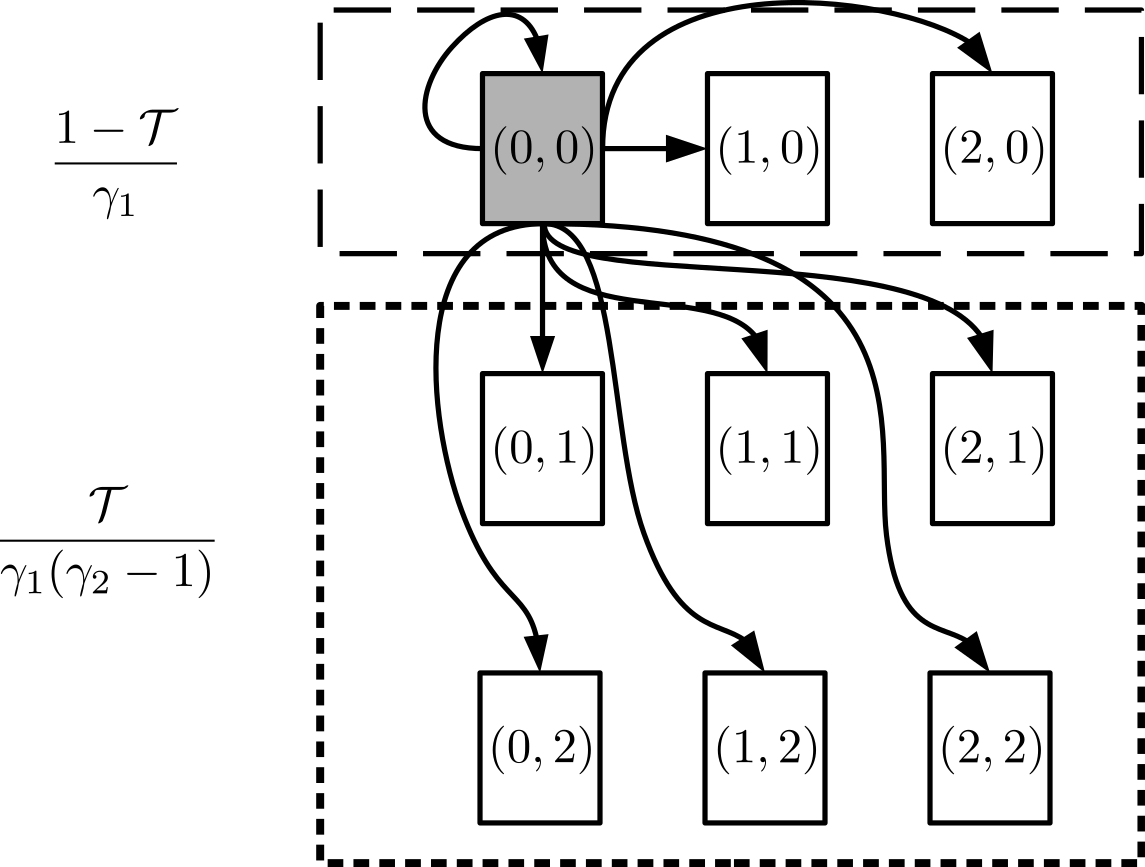}
		\caption{ \footnotesize Coupling from the VN perspective of section $(0,0)$ with $\gamma_1=\gamma_2=3$. The probability of connecting to a section in the top box is $\frac{1-\mathcal{T}}{\gamma_1}$ and of connecting to a section in the bottom box is $\frac{\mathcal{T}}{\gamma_1(\gamma_2-1)}$.\vspace{-0.5cm}}
		\label{fig_coupling}
	\end{figure}
	
	\begin{figure*}[ht]
		\small
		\begin{equation}\label{eq_bp_forward_de}
		\begin{split}
		y^{(l+1)}_{(i,j)} &= 1-(1-\frac{1-\mathcal{T}}{\gamma_1}\sum_{k=0}^{\gamma_1-1} x^{(l)}_{(i-k,j)}-\frac{\mathcal{T}}{\gamma_1(\gamma_2-1)}\sum_{k=0}^{\gamma_1-1}\sum_{r=1}^{\gamma_2-1} x^{(l)}_{(i-k,(j-r)_{L_2})})^{d_r-1} \\
		x^{(l+1)}_{(i,j)} &= 
		\epsilon(\frac{1-\mathcal{T}}{\gamma_1}\sum_{k=0}^{\gamma_1-1}y^{(l+1)}_{(i+k,j)} + \frac{\mathcal{T}}{\gamma_1(\gamma_2-1)}\sum_{k=0}^{\gamma_1-1}\sum_{r=1}^{\gamma_2-1}y^{(l+1)}_{(i+k,(j+r)_{L_2})})^{d_l-1}
		\end{split}
		\end{equation}
		\vspace{-0.7cm}
	\end{figure*}
	
	In this section, we define an MD-SC code ensemble  $\mathcal{C}_{\text{MD}}=\mathcal{C}(d_l,d_r,L_1,\gamma_1,L_2,\gamma_2,\mathcal{T})$. The parameters $d_l$ and $d_r$ denote the degrees of the variable nodes (VNs) and check nodes (CNs), respectively. We denote $L_1$ as the \textit{1-dimensional (1D)  coupling length}  and $L_2$ as the \textit{MD coupling length}. Additionally, we define $1 \leq \gamma_1 \leq L_1$ as the \textit{1D coupling depth} and $1 \leq \gamma_2\leq L_2$ as the \textit{MD coupling depth}, which specify the coupling distance along a dimension. We refer to $0\leq \mathcal{T} \leq 1$ as the \textit{density} of the edges for the coupling along the second dimension. We note that for $\mathcal{T} = 0$, this ensemble would degenerate into $L_2$ uncoupled SC-LDPC codes. 
	
	First, we define the building blocks of our construction. We denote $M$ as the \textit{section size}. A \textit{section} is a collection of $M$ VNs and $M(d_l/d_r)$ CNs and is represented by a tuple $(i,j) \in \mathbb{Z}^2$. The code is constructed by extracting only the VNs in sections $[(L_1,L_2)] \triangleq  \{(j,k): j \in [L_1], k\in [L_2]\} \subset \mathbb{Z}^2$. CNs that are not connected to the VNs in $[(L_1,L_2)]$ after coupling are purged from the code. We will describe the coupling shortly. The VNs and CNs in section $(i,j)$ make up the $i^{th}$ position of the $j^{th}$ segment of the overall MD-SC code. For convenience, we denote $(i,\cdot)$ as the $i^{th}$ position of the code and $(\cdot,j)$ as the $j^{th}$ segment of the code.

	Now, we describe coupling of the sections. For each of the $d_l$ edges incident to a VN in section $(i,j)$, we flip a biased coin with probability of heads being $\mathcal{T}$. If it is tails, we choose a section uniformly and independently from $\{(i+k,j): \; k \in [\gamma_1]\}$, \textcolor{black}{ and if it is heads}, we choose a section uniformly and independently from  $\{(i+k,(j+r)_{L_2}): \; k \in [\gamma_1], r\in [\gamma_2]\setminus\{0\}\}$. After choosing a section to connect to, a CN is picked uniformly at random from the $M(d_l/d_r)$ CNs in that section to connect the edge to. This coupling can also be viewed from the CN perspective. In other words, each of the $d_r$ edges of a CN in section $(i,j)$ is uniformly connected to a VN from sections $\{(i-k,j): \; k \in [\gamma_1]\}$ with probability $1-\mathcal{T}$ or sections $\{(i-k,(j-r)_{L_2}): \; k \in [\gamma_1], r\in [\gamma_2]\setminus\{0\}\}$ with probability $\mathcal{T}$. Coupling is performed such that no parallel edges are formed. An example of this coupling is in Fig.~\ref{fig_coupling}. \footnote{\textcolor{black}{Note that the  first dimension is terminated and the second dimension is circularly coupled, as it is also used in previous works  \cite{Esfahanizadeh2019MultiDimensionalSC,Hareedy2019MinimizingTN}. For our work, we utilize this coupling to restrict the direction of the decoding wave to simplify our decoder analysis.}}
	
	\vspace{-0.0cm}
	\subsection{Density Evolution}
	We analyze the performance of our ensemble for the binary erasure channel with erasure probability $\epsilon$ (BEC($\epsilon$)) under BP decoding. By taking $M\to \infty$, we use DE techniques \cite{RichardsonModernCodingTheory,KudekarIT2011} to define BP DE equations to analyze  our ensemble.
	
	Let $x_{(i,j)}$ and $y_{(i,j)}$ denote the erasure probability of an outgoing message from a VN and CN in section $(i,j)$, respectively. We define $\underline{x} = \{x_{(i,j)}\}$ as the constellation of VN erasure probabilities. We initialize the constellation with $x^{(0)}_{(i,j)} = 1$ for $(i,j) \in [(L_1,L_2)]$ and $0$ otherwise. 
	According to our construction method, the resulting BP DE equations are given in \cref{eq_bp_forward_de}. For this work, we employ a \textit{flooding schedule} where all CN messages are updated before updating the VN messages, and vice versa.
	
	For convenience, we write the \textit{update} purely in terms of $x_{(i,j)}$ as 
	$	x^{(l+1)}_{(i,j)} = f(\{x^{(l)}_{(i\pm k ,(j \pm r)_{L_2})}: \; (k,r) \in [(\gamma_1,\gamma_2)] \}) 
	$
	for $(i,j) \in [(L_1,L_2)]$. It can be verified that $f(\cdot)$ is monotonic in each of its arguments. Given a target erasure probability $\delta$, we define the BP threshold as $\epsilon^*_{\delta}$ such that for $\epsilon \leq \epsilon^{*}_{\delta}$ BP is able to decode all bits to at most a target erasure probability $\delta$ after an infinite number of iterations.
	
	\vspace{-0.15cm}
	\section{MD-SC-LDPC Ensemble Analysis}
	
	Before describing our decoder, it is important to understand the features resulting from the MD-SC code structure. 
	The following lemma shows that the $\mathcal{C}(d_l,d_r,L_1,\gamma_1,L_2,\gamma_2,\mathcal{T})$ ensemble has the same asymptotic properties (design rate and BP threshold) as the standard 1D-SC code ensemble $\mathcal{C}_{\text{1D}} = \mathcal{C}(d_l,d_r,L_1,\gamma_1)$ (see \cite{KudekarIT2011} for full description of $\mathcal{C}_{\text{1D}}$). A similar lemma was introduced for the MD-SC ensemble defined in \cite{OhashiISIT2013}, and we extend the concept for our new ensemble definition that incorporates the coupling density.
	\begin{lemma}\label{Lemma_eq}
		Let $\epsilon^{*}_{\delta}(\mathcal{C})$ and $R(\mathcal{C})$  refer to the BP threshold and design rate of a code ensemble C, respectively. Then,
		\begin{align}
		\epsilon^{*}_{\delta}(\mathcal{C}(d_l,d_r,L_1,\gamma_1,L_2,\gamma_2,\mathcal{T})) = \epsilon^{*}_{\delta}(\mathcal{C}(d_l,d_r,L_1,\gamma_1)) \label{bp_equal}\\
		R(\mathcal{C}(d_l,d_r,L_1,\gamma_1,L_2,\gamma_2,\mathcal{T})) = R(\mathcal{C}(d_l,d_r,L_1,\gamma_1)). \label{rate_equal}
		\end{align}
	\end{lemma}

		\begin{proof}
			To prove equivalency of the BP threshold, we show that  $x^{(l)}_{(i,j)}= \bar{x}^{(l)}_{i}$ for $l\geq 0$ and $ j\in [L_2]$ where $\bar{x}^{(l)}_{i}$ is the BP DE for $\mathcal{C}(d_l,d_r,L,\gamma)$. Therefore, the limits of these erasure probabilities will be the same which guarantees the same threshold.
			
			We prove this claim by induction. Clearly, $x^{(0)}_{(i,j)}= \bar{x}^{(0)}_{i} = 1$. Now, assume $x^{(l)}_{(i,j)}= \bar{x}^{(l)}_{i}$ holds true. Then,
			\begin{equation*}
			\begin{split}
			y^{(l+1)}_{(i,j)} &=  1-(1-\frac{1-\mathcal{T}}{\gamma_1}\sum_{k=0}^{\gamma_1-1} x^{(l)}_{(i-k,j)}\\
			&\quad -\frac{\mathcal{T}}{\gamma_1(\gamma_2-1)}\sum_{k=0}^{\gamma_1-1}\sum_{r=1}^{\gamma_2-1} x^{(l)}_{(i-k,(j-r)_{L_2})})^{d_r-1} \\
			&= 1-(1-\frac{1}{\gamma_1}\sum_{k\in[\gamma_1]} \bar{x}^{(l)}_{i-k})^{d_r-1} = \bar{y}^{(l+1)}_{i} \\
			\end{split}
			\end{equation*}
			\begin{equation*}
			\begin{split}
			x^{(l+1)}_{(i,j)}  &= 
			\epsilon(\frac{1-\mathcal{T}}{\gamma_1}\sum_{k=0}^{\gamma_1-1}y^{(l+1)}_{(i+k,j)} \\
			&\quad + \frac{\mathcal{T}}{\gamma_1(\gamma_2-1)}\sum_{k=0}^{\gamma_1-1}\sum_{r=1}^{\gamma_2-1}y^{(l+1)}_{(i+k,(j+r)_{L_2})})^{d_l-1}\\
			&= \epsilon(\frac{1}{\gamma_1}\sum_{k\in[\gamma_1]}\bar{y}^{(l+1)}_{i+k})^{d_l-1} =  \bar{x}^{(l+1)}_{i},
			\end{split}
			\end{equation*}
			for $(i,j) \in [(L_1,L_2)]$ where $\bar{y}^{(l+1)}_{i}$ is the check-to-variable messages for $C_{\text{1D}}$ . Thus, $x^{(l)}_{(i,j)}= \bar{x}^{(l)}_{i}$ for $l\geq 0$ which proves (\ref{bp_equal}).
			
			To prove (\ref{rate_equal}), we observe that the rate-loss is due to the coupling along the first dimension since the coupling along the second dimension wraps around. As such, the expected number of disconnected CNs in section $(i,j)$ for $i\in [\gamma_1]$ is $M\frac{d_l}{d_r}(\frac{\gamma_1 -1 -i}{\gamma_1})^{d_r}$. By symmetry, CNs in sections $(i,j)$ for $i\in \{L_1+k: k \in [\gamma_1]\}$ have the same expected number of disconnected CNs. Finally, CNs for sections $(i,j)$ for $i\in \{\gamma_1,\gamma_1+1,\dots,L_1-1\}$ have zero expected disconnected CNs. Hence, we get 
			\begin{align*}
			&R(C(d_l,d_r,L_1,\gamma_1,L_2,\gamma_2)) \\
			&= 1-\frac{(ML_2\frac{d_l}{d_r})(L_1+\gamma_1-1 - 2\sum^{\gamma_1-1}_{i=0}(\frac{i}{\gamma_1})^{d_r})}{ML_2L_1} \\
			&= 1 - \frac{d_l}{d_r}(1+\frac{\gamma_1-1 - 2\sum^{\gamma_1-1}_{i=0}(\frac{i}{\gamma_1})^{d_r}}{L_1})
			\end{align*}
			which is the rate for $C(d_l,d_r,L_1,\gamma_1)$.
		\end{proof}
	
	While our MD-SC ensemble and the standard SC ensemble have exactly the same asymptotic properties, \textcolor{black}{they may differ in their finite-length performances.} To demonstrate the differences, we analyze the occurrence probability of a size-$2$ stopping set for VNs within a section. A size-$k$ stopping set is a subset of $k$ VNs where all neighboring CNs of this subset connect to the subset at least twice \cite{RichardsonModernCodingTheory}. If all VNs in a stopping set are erased, the BP decoder fails to decode this set of VNs.
	
	To see the effect of $\gamma_2$ and $\mathcal{T}$, we calculate the probability of a size-$2$ stopping set occurring for two VNs in the same section. This probability acts as a rough upper bound on the probability of size-$2$ stopping sets for any pair of VNs. The following lemma is inspired by \cite{Kudekar2015WaveLikeSO} where their analysis is performed for 1D-SC codes.
	

	\begin{lemma}\label{lemma_2}
	Assume $\gamma_1M\frac{d_l}{d_r} > d_l$. Consider the ensemble $C(d_l,d_r,L_1,\gamma_1,L_2,\gamma_2,\mathcal{T})$. Given two VNs in the same section, the probability that they form a stopping set is 
	\textcolor{black}{
	\begin{equation}\label{eq_stop_p_stop}
	P_{stop} = \sum_{a,b \geq 0 : a+b = d_l}\frac{(1-\mathcal{T})^{2a}\mathcal{T}^{2b}\binom{d_l}{a}^2(1-\frac{1}{d_r})^{d_l}}{ \sum_{l=0}^{a}\sum_{k=0}^{b}\mathcal{K}^{a,b}_{l,k}(1-\frac{1}{d_r})^{l+k}},
	\end{equation}}
	where 
	\begin{equation}\label{eq_stop_k_ab}
	\small
	\mathcal{K}^{a,b}_{l,k}= \binom{a}{l}\binom{b}{k}\binom{\gamma_1M\frac{d_l}{d_r}-a}{a-l}\binom{\gamma_1(\gamma_2-1)M\frac{d_l}{d_r}-b}{b - k}.
	\end{equation}
	\end{lemma}

		\begin{proof}
			Consider two VNs $v_1$ and $v_2$ in section $(i,j)$. To form a stopping set, all of their edges must connect to the same set of CNs. First, we note that the VNs connect to CNs in positions $\{(i+k,j): \; k \in [\gamma_1]\}$ with probability $1-\mathcal{T}$ and to CNs in positions $\{(i+k,(j+r)_{L_2}: \; k \in [\gamma_1],  r\in [1,2,\dots,\gamma_2-1]\}$ with probability $\mathcal{T}$. We denote the first set as $S_0$ and the second set as $S_1$. Out of the $d_l$ edges of node $v_1$, the probability that $a$ edges connect to $S_0$ and $b$ edges connect to $S_1$ is $(1-\mathcal{T})^{a}\mathcal{T}^{b}\binom{d_l}{a}$. To make a stopping set, $v_2$ must also have $a$ edges connected to $S_0$ and $b$ edges connected to $S_1$. Thus,
			\begin{equation*}
			P_{stop} = \sum_{a,b \geq 0 : a+b = d_l}(1-\mathcal{T})^{2a}\mathcal{T}^{2b}\binom{d_l}{a}^2\mathcal{P}_{a,b}
			\end{equation*}
			where $\mathcal{P}_{a,b}$ is the conditional probability of $\{v_1,v_2\}$ being a stopping set given that both VNs have $a$ edges connected to $S_0$ and $b$ edges connected to $S_1$.
			
			To find $\mathcal{P}_{a,b}$, it is necessary to find the probability that $v_1$ and $v_2$ connect to the same CNs in $S_0$ and $S_1$. We remind that an edge of a VN node $v_1$ or $v_2$ is equally likely to connect to any CN within the sets. As such, we can use a counting argument to calculate the conditional probability of a stopping set within that subset of sections.  Recall that there are no parallel edges.  Each CN has $d_r$ sockets for an edge to be connected to. We can fix $N(v_1)$ since all the subsets of CNs that $v_1$ can be connected to is equally likely. Let $K^{c}_{stop,S}$ be the number of choices for set $S$ where $v_1$ and $v_2$ connect to the same CNs with $c$ edges and let $K^{c}_{S}$ be the total number of choices for $S$ with $c$ edges from each VN. Then, we have $\mathcal{P}_{a,b} =  \frac{K^{a}_{stop,S_0}}{K^{a}_{S_0}} \cdot \frac{K^{b}_{stop,S_1}}{K^{b}_{S_1}} $.  
			
			Let us consider the number of choices for $S_0$. Thus, 
			\begin{equation}\label{eq_stop_1}
			K^{a}_{stop,S_0} = a!(d_r-1)^{a}
			\end{equation}
			where $a!$ is due to the permutation of edges and  $(d_r-1)^{a}$ is the number of different ways of connecting the edges of $v_2$ to the free $d_r-1$ sockets of $N(v_1)$. To calculate $K^{a}_{S_0}$, we note that in general $v_1$ and $v_2$ can have $l$ common CN neighbor in $S_0$ with $0\leq l \leq a$. There are $\binom{a}{l}(d_r-1)^{l}$ socket selections for the $l$ common CNs. On the other hand, there are $\binom{\gamma_1M\frac{d_l}{d_r}-a}{a-l}(d_r)^{a-l}$ socket selections for the other CNs. Thus,
			\begin{equation}\label{eq_stop_2}
			K^{a}_{S_0} = a!\sum_{l=0}^{a}\binom{a}{l}\binom{\gamma_1M\frac{d_l}{d_r}-a}{a-l}(d_r-1)^{l}(d_r)^{a-l}.
			\end{equation}
			
			We can get similar results for $S_1$ except that the total number of CNs in $S_1$ is $\gamma_1(\gamma_2-1)M\frac{d_l}{d_r}$. As such,
			\begin{equation}\label{eq_stop_3}
			K^{b}_{stop,S_1} = b!(d_r-1)^{b} 
			\end{equation}
			and
			\begin{equation}\label{eq_stop_4}
			K^{b}_{S_1} = b!\sum_{k=0}^{b}\binom{b}{k}\binom{\gamma_1(\gamma_2-1)M\frac{d_l}{d_r}-b}{b-k}(d_r-1)^{k}(d_r)^{b-k}.
			\end{equation}
			
		\textcolor{black}{Combining \Crefrange{eq_stop_1}{eq_stop_4} into $\mathcal{P}_{a,b}$, we get the simplified \cref{eq_stop_p_stop} and \cref{eq_stop_k_ab}.}

		\end{proof}
	
	\begin{figure}[t]
		\centering
		\includegraphics[height = 5.5cm]{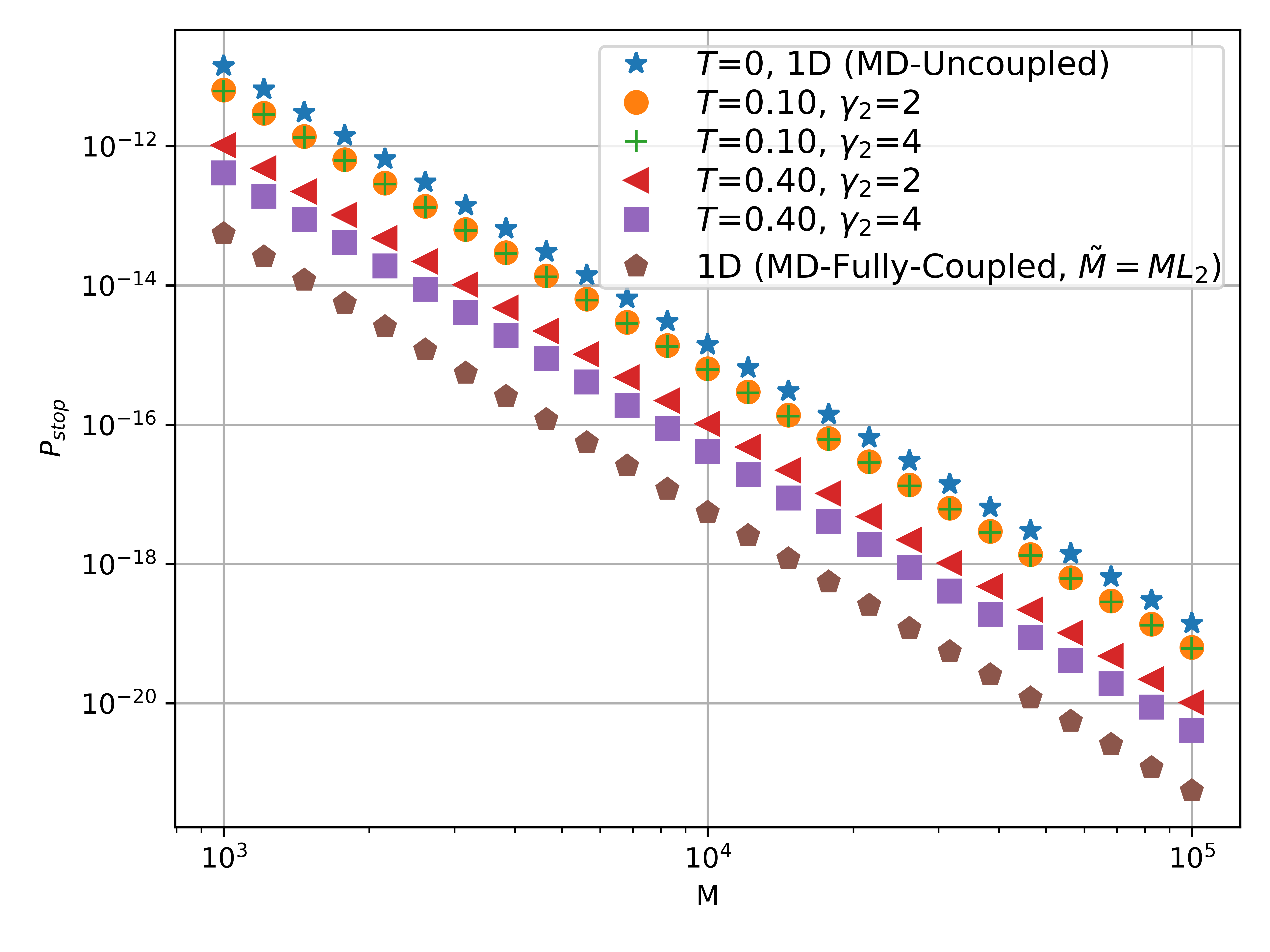}
		\vspace{-0.3cm}
		\caption{ \footnotesize \textcolor{black}{ We compare $P_{stop}$ for 1D and MD codes versus $M$ and for code parameters $d_l=4,d_r=8,L_2 = 3$, and $\gamma_1=2$. For the 1D codes, the MD-Uncoupled case is when $\mathcal{T} =0$, i.e., there are $L_2$ 1D coupled codes, and the MD-Fully-Coupled case is when $\gamma_2=L_2$ and $\mathcal{T} = \frac{\gamma_2-1}{\gamma_1}$ which  is equivalent to a 1D code with section size $\tilde{M} = ML_2$. \vspace{-0.5cm}}}
		\label{fig_stop}
		
	\end{figure}
	
	Fig.~\ref{fig_stop} compares $P_{stop}$ for relevant SC codes \cite{Aref2016FiniteLengthAO} and various MD-SC codes. We observe that increasing $\mathcal{T}$ and/or $\gamma_2$ results in a decrease for $P_{stop}$ and $\gamma_2$ is less influential in this regard for small values of $\mathcal{T}$. While this analysis is fairly coarse, it suggests that the finite-length performance of the MD-SC ensemble is improved by increasing $\gamma_2$ and $\mathcal{T}$. However, in subsequent sections, we demonstrate the benefit of small $\gamma_2$ and $\mathcal{T}$ for the windowed decoding threshold. Thus, there is a trade-off between finite-length and asymptotic performance that need to be considered in the design. The rest of the paper will focus purely on the asymptotic characteristics of designing a windowed decoder. 
	
	
	\section{Non-Uniform Windowed Decoding}
	In this section, we describe a general approach to perform non-uniform windowed decoding  on the $C(d_l,d_r,L_1,\gamma_1,L_2,\gamma_2,\mathcal{T})$ ensemble. We define a subset of VNs for which BP will be performed over as a \textit{window configuration (WC)}. Every WC has a unique set of VNs that are aimed to be decoded, called the \textit{targeted VNs (TVNs)}. The TVNs of each WC are VNs of a single section $(i,j)$ of the code. We denote $\underline{x}_{\{t\}}$ as the global constellation after $t$ WCs have been processed. The initial constellation $\underline{x}_{\{0\}}$ is set to $x_{(i,j),\{0\}} = 1$ for $(i,j) \in [(L_1,L_2)]$ and $0$ otherwise. 
	
	Assume the VNs of section $(i_t,j_t)$ are the TVNs after $t$ WCs are processed. We denote $\mathcal{W} = [W_0,W_1,\dots,W_{L_2-1}]$ to be the \textit{vector of window sizes} of the WCs. Given $(i_t,j_t)$, we define $S^{\mathcal{W}}_{(i_t,j_t)} = \{(i_t+k,j_t+r): \; r \in [L_2], k \in [W_{(j_t+r)_{L_2}}]\}$ as the WC over which BP will be performed. For any specific WC, the window sizes are cyclically shifted so that $W_0$ is centered on the TVNs. 
	
	We define $\underline{z}$ as the window constellation.  We initialize $\underline{z}^{(0)}$ by the current global constellation, i.e., $\underline{z}^{(0)} = \underline{x}_{\{t\}} $. We then update $\underline{z}^{(l)}$ by
	\begin{equation} \label{eq_windowed_update}
	z^{(l+1)}_{(i,j)} = \begin{dcases}
	z^{(l)}_{(i,j)} \text{, if } (i,j) \notin S^{\mathcal{W}}_{(i_t,j_t)} \\
	f(\{z^{(l)}_{(i\pm k ,(j \pm r)_{L_2})}: \; (k,r) \in [(\gamma_1,\gamma_2)] \}) \text{, else}
	\end{dcases}
	\end{equation}
	for $l \in [I_{(i_t,j_t)}]$, where $I_{(i_t,j_t)}$ is the maximum number of iterations and is chosen such that $z^{I_{(i_t,j_t)}}_{(i_t,j_t)} \leq \delta$ for the \textit{target erasure probability} $\delta$ \footnote{\textcolor{black}{A \textit{finite} number of needed iterations is achievable when $\epsilon$ is sufficiently smaller than the threshold $ \epsilon^{*}_{\delta}$. }}.  After $I_{(i_t,j_t)}$ iterations, the global constellation is updated by
	\begin{equation*}
	x_{(i,j),\{t+1\}} =  \begin{dcases*}
	z^{I_{(i_t,j_t)}}_{(i,j)} & if $ (i,j) = (i_t,j_t)$ \\
	x_{(i,j),\{t\}} & if $(i,j) \neq (i_t,j_t)$. \\
	\end{dcases*}
	\end{equation*}
	If no WC is repeated, all the sections are updated after $L_1L_2$ WCs are processed. We define $\epsilon^\textnormal{WC}_{\delta,\mathcal{W}}$ as the BP threshold such that for $\epsilon \leq \epsilon^\textnormal{WC}_{\delta,\mathcal{W}}$ the non-uniform windowed decoder is able to decode all TVNs to at most a target erasure probability $\delta$. An example of a WC is presented in Fig.~\ref{fig_window}.
	
	\begin{figure}[tb]
		\centering
		\includegraphics[scale=0.4]{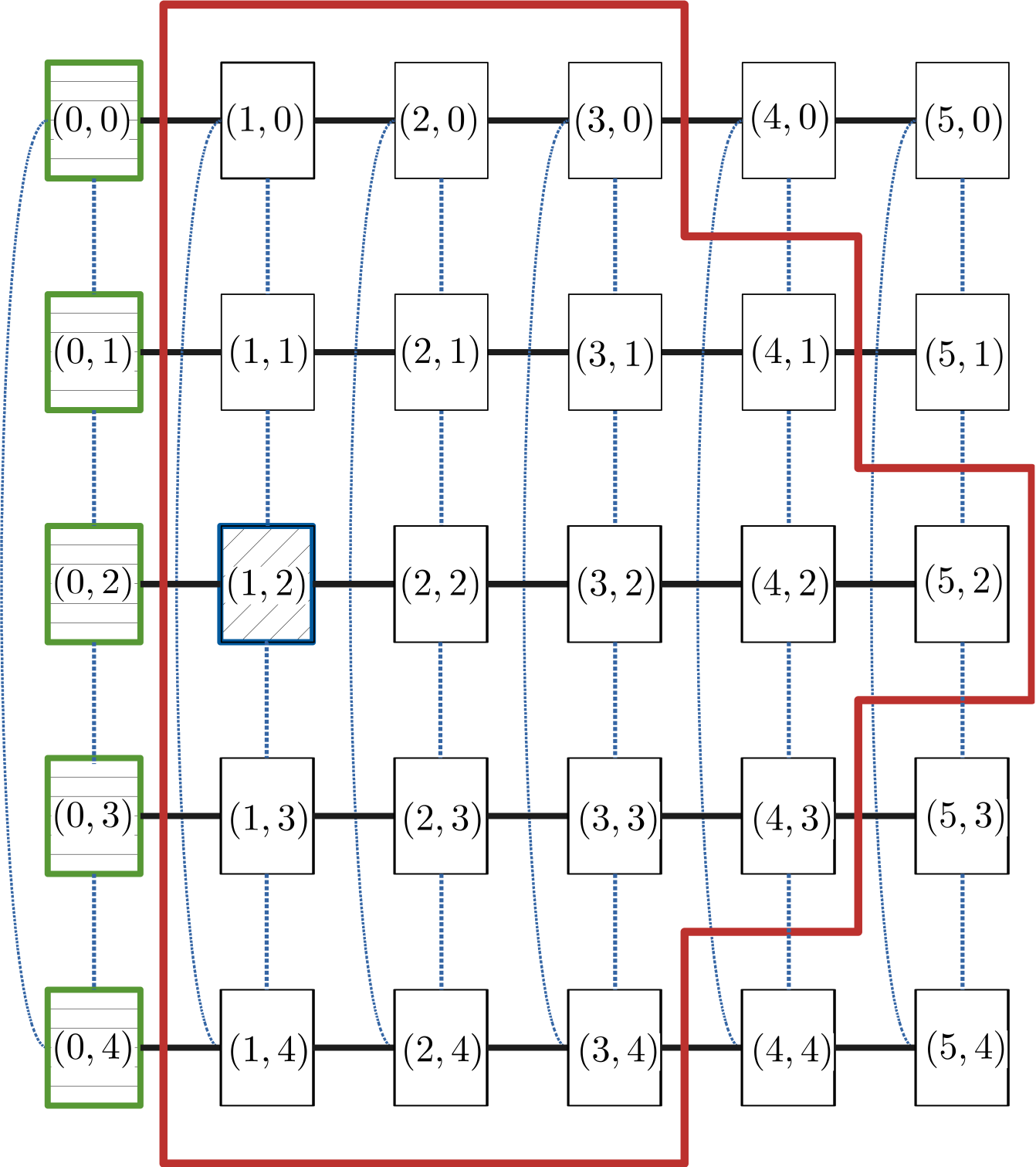}
		\caption{\footnotesize Example of a non-uniform WC with TVNs in section $(1,2)$ and window sizes $[5,4,3,3,4]$. The green rectangles with horizontal hatching represent the decoded VNs and the blue rectangle with diagonal hatching represents the targeted VNs. Due to the TVNs being in segment $2$, the window is shifted so that window size $W_0$ is used for segment $2$.\vspace{-0.3cm}}
		\label{fig_window}
	\end{figure}

	We briefly analyze the complexity and latency of the decoder. The complexity of the $t^{\text{th}}$ WC is $\mathcal{O}(\sum_{i=0}^{L_2-1}W_i I_{(i_t,j_t)})$ since the number of VNs in a WC is $\mathcal{O}(\sum_{i=0}^{L_2-1}W_i)$ and all these VNs are updated for $I_{(i_t,j_t)}$ iterations. For convenience, we denote $s(\mathcal{W}) = \sum_{i=0}^{L_2-1}W_i$. Additionally, the number of VNs that need to be accessed to process a WC is at most $\mathcal{O}(s(\mathcal{W}))$. As such, the latency of a WC is upper bounded by $\mathcal{O}(s(\mathcal{W})) + \mathcal{O}(s(\mathcal{W})I_{(i_t,j_t)}) = \mathcal{O}(s(\mathcal{W})I_{(i_t,j_t)})$. For a block BP decoder, the latency is $\mathcal{O}(L_1L_2I_{BP})$ where $I_{BP}$ represents the number of iterations. Even if the iteration number is the same for both (in general, $I_{(i_t,j_t)} \leq I_{BP}$), latency is reduced by at least a factor of $\mathcal{O}(\frac{s(\mathcal{W})}{L_1L_2})$. As such, we denote the latency or complexity constraint as $s(\mathcal{W}) \leq C$ for some integer $C$. We denote $C$ as the \textit{window complexity}.
	
	We note that an equivalent performance to uniform windowed decoding \cite{IyengarIT2013} can be achieved by setting $ W_i = W_j$ for $i\neq j$. Thus, the latency is $\mathcal{O}(\frac{L_2W}{L_1L_2}) = \mathcal{O}(\frac{W}{L_1})$ and the decoding threshold is the same as for a windowed decoder of SC codes by the same rationale as shown in \Cref{Lemma_eq}. Furthermore, the performance of a uniform windowed decoder is independent of $\gamma_2$ and $\mathcal{T}$. As such, we consider the uniform decoder for baseline performance and will demonstrate how allowing for non-uniform $\mathcal{W}$ results in a finer control of decoder complexity, latency, and reliability. 
	
	Given the general construction, we address three design degrees of freedom in the rest of this section:
	\begin{enumerate}
		\item What order should the WCs be processed in?
		\item What is the best $\mathcal{W}$ given $s(\mathcal{W})$?
		\item What should the number of iterations be set to for each WC?
	\end{enumerate}
	
	It is clear that the performance of the decoder jointly depends on the previous three questions. In the subsequent subsections, we answer these questions and motivate our choices.

	\subsection{Processing Order}
	
	We remind in \Cref{Lemma_eq}, we proved that the thresholds of $\mathcal{C}_{\text{MD}}$ and $\mathcal{C}_{\text{1D}}$ are equivalent because the erasure probabilities for both BP DEs exactly track each other. This implies that the decoding wave \cite{Kudekar2015WaveLikeSO} also appears for $\mathcal{C}_{\text{MD}}$ and that it travels \textcolor{black}{along the first dimension of the code.}
	
	
	This observation implies we should process the WCs along the first dimension to follow the decoding wave. As such, we impose the constraint that no section in $(i,\cdot)$ can be processed before a section in $\{(k,\cdot) : k <i\}$. Therefore, we only need to choose the processing order of the sections $(i,\cdot)$. One intuitively reasonable choice is to process them in the order $0,1,\dots,L_2-1$ which we call the \textit{natural order}. For this processing order, the next TVNs that will be processed are the ones closest to the previously decoded TVNs which help the most in decoding the new TVNs. \textcolor{black}{In the simulations, we demonstrate that ordering has a strong effect on the finite number of iterations.}
	
	\subsection{Design of Window Sizes}
	
	Now that a WC processing order is settled, we identify the best choice of $\mathcal{W}$ for this ordering. In order to choose a $\mathcal{W}$ independent of $L_1$, we analyze a WC whose performance is a lower bound on the performances of all WCs. 
	
\vspace{-0.5cm}
	\textcolor{black}{
		\begin{definition} (Worst-Case WC) Given $\mathcal{W}$, we define $\underline{q}^{(l)}$ to represent the \textit{worst-case window constellation} where $q^{(0)}_{(i,j)} = \delta$  if $i<0$ and $1$ otherwise.
		We then update $\underline{q}^{(l)}$ by \cref{eq_windowed_update} with the TVNs designated in section $(0,0)$.
		We also define the worst-case window BP threshold as
		\begin{equation*}
		\epsilon^{worst}_{\delta,\mathcal{W}} = \sup\{\epsilon > 0 : \; q^{(l)}_{0,0} \to a \text{ as } l \to \infty \; \text{s.t.} \, a\leq \delta\},
		\end{equation*}
		since every WC aims at decoding the targeted VNs.
		\end{definition}
	}
	\begin{claim}
		For all $\epsilon \leq \epsilon^{worst}_{\delta,\mathcal{W}}$, the non-uniform windowed decoder is able to decode the VNs of all sections to an erasure probability at most the target erasure probability $\delta$.
	\end{claim}

	To understand this claim, recall that no section in $(i,\cdot)$ can be processed before a section in $\{(k,\cdot) : k <i\}$. For a section in $(i,\cdot)$, we can assume the erasure probabilities of VNs in sections $\{(k,\cdot) : k <i\}$ are at most $\delta$. Therefore, the first TVNs processed have the least help from the other sections since they have yet to be decoded. By monotonicity of $f(\cdot)$, the Worst-Case WC DEs dominate the DEs for any WC in the code. Hence, if the Worst-Case WC decodes its TVNs to an erasure probability at most $\delta$, then so do all the WCs.
	
	We intend to find a $\mathcal{W}$ that satisfies $s(\mathcal{W}) \leq C$ and maximizes $\epsilon^{worst}_{\delta,\mathcal{W}}$. Intuitively, one may think a $\mathcal{W}$ which satisfies $s(\mathcal{W}) < C$ must perform worse than those that satisfy $s(\mathcal{W}) = C$. However, this is not always true. The following lemma provides an ordering to the performances of different choices of $\mathcal{W}$.
	\begin{lemma}
		Given distinct $\mathcal{W}$ and $\mathcal{W}^{\prime}$ where $\mathcal{W}^{\prime} 	\preceq \mathcal{W} $, the following inequality holds for the worst-case window thresholds
		\begin{equation*}
		\epsilon^{worst}_{\delta,\mathcal{W}^{\prime}} \leq \epsilon^{worst}_{\delta,\mathcal{W}}.
		\end{equation*}
	\end{lemma}
	\begin{proof}
		Let $\underline{z}^{(l)}_{\mathcal{W}}$ and $\underline{z}^{(l)}_{\mathcal{W}^\prime}$ be the worst-case window constellations for WCs with $\mathcal{W}$ and $\mathcal{W}^\prime$ at iteration $l$, respectively. Recall that $S^{\mathcal{W}}$ and $S^{\mathcal{W}^\prime}$ are the WCs over which BP will be performed on. We note that $S^{\mathcal{W}^\prime} \subset S^{\mathcal{W}}$. By definition, $\underline{z}^{(0)}_{\mathcal{W}} = \underline{z}^{(0)}_{\mathcal{W}^\prime}$. Consider a section $(i,j)$ such that $(i,j)\in S^{\mathcal{W}} \setminus S^{\mathcal{W}^\prime}$. As such, 
		$
		z^{(1)}_{(i,j),\mathcal{W}} = \epsilon \leq 1 = z^{(1)}_{(i,j),\mathcal{W}^{\prime}}.
		$
		Now, consider a section $(i,j)$ such that $(i,j)\in S^{\mathcal{W}} \cap S^{\mathcal{W}^\prime}$. We have $z^{(1)}_{(i,j),\mathcal{W}} = z^{(1)}_{(i,j),\mathcal{W}^\prime} = \epsilon$. Hence, we  conclude $\underline{z}^{(1)}_{\mathcal{W}} \preceq \underline{z}^{(1)}_{\mathcal{W}^\prime}$.
		By induction on $l$ and monotonicity of $f(\cdot)$, we conclude $\underline{z}^{(\infty)}_{\mathcal{W}} \preceq \underline{z}^{(\infty)}_{\mathcal{W}^\prime}$ which completes the proof.
	\end{proof}
	
	Thus, for every $\mathcal{W}^{\prime}$ that satisfies $s(\mathcal{W}^{\prime}) < C$, there exists a $\mathcal{W}$ such that $s(\mathcal{W}) = C$ and has a better threshold than $\mathcal{W}^{\prime}$. Hence, we can restrict our attention to all choices of $\mathcal{W}$ that satisfy $s(\mathcal{W}) = C$.

	\subsection{Iteration per Window}
	With the processing order and $\mathcal{W}$ fixed, we calculate the minimum $I_{(i_t,j_t)}$ for each section to guarantee the target erasure probability $\delta$ is met. In the next section, we provide simulations on how the iteration number changes as function of the ensemble parameters.

	\section{Simulations}
	In this section, we demonstrate through simulations the flexibility and improvements offered by our non-uniform windowed decoder. Additionally, we empirically justify  our design choices, e.g., the use of worst-case WC and the processing order, and show that such design choices result in a superior performance.
	
	\subsection{Worst-Case WC Analysis for Decoder Design}

	\begin{table}[t]
		\centering
		\caption{\footnotesize Thresholds for Window Sizes with the largest $\epsilon^{worst}_{\delta,\mathcal{W}}$ for $d_l = 4$, $d_r = 8$, $\gamma_1=2$, and $\delta= 10^{-12}$. Window sizes are constrained between $2$ and $7$, for $L_2=7$, and between $2$ and $5$, for $L_2=9$.}
		\setlength\tabcolsep{4.7pt}
		\begin{tabular}{c|c|c|c|c|c|c}
			$L_2$&$\gamma_2$& $\mathcal{T}$ & $C$ & $\mathcal{W}$ & $\epsilon^{worst}_{\delta,\mathcal{W}}$ & $\epsilon^\textnormal{WC}_{\delta,\mathcal{W}}$  \\
			\hline
			
			$7$ & $2$ & $0.05$ & $28$ & $(5, 5, 4, 2, 3, 4, 5)$& $\approx 0.4829$ &$\approx 0.4829$\\	
			
			$7$ &$2$ & $0.1$ &$28$ & $(5, 5, 4, 3, 3, 4, 4)$& $\approx 0.4722$ & $\approx 0.4722$\\
			
			$7$ &$3$ & $0.05$ & $28$ & $(5, 4, 4, 3, 4, 4, 4)$& $\approx 0.4723$ & $\approx 0.4723$\\	
			
			$7$ &$3$ & $0.1$ &$28$ & $(4, 4, 4, 4, 4, 4, 4)$& $\approx 0.4685$ & $\approx 0.4685$\\

			$9$ & $2$ & $0.05$ & $36$ & $(5, 5, 4, 3, 2, 3, 4, 5, 5)$& $\approx 0.4872$ &$\approx 0.4872$\\	
			
			$9$ &$2$ & $0.1$ &$36$ & $(5, 5, 4, 3, 2, 3, 4, 5, 5)$& $\approx 0.4806$ & $\approx 0.4806$\\

			$9$ &$3$ & $0.05$ & $36$ & $(5, 5, 5, 2, 3, 3, 4, 4, 5)$& $\approx 0.4767$ & $\approx 0.4767$\\	
			
			$9$ &$3$ & $0.1$ &$36$ & $(4, 4, 4, 4, 4, 4, 4, 4, 4)$& $\approx 0.4685$ & $\approx 0.4685$\\
			\hline
		\end{tabular}
		\label{Table_worst_window}
		\vspace{-0.3cm}
	\end{table}

	To justify the use of Worst-Case WC to design the window sizes, we compare  $\epsilon^\textnormal{WC}_{\delta,\mathcal{W}}$ and $\epsilon^{worst}_{\delta,\mathcal{W}}$ for several cases. In \Cref{Table_worst_window}, we show the thresholds for different window sizes that were chosen to maximize $\epsilon^{worst}_{\delta,\mathcal{W}}$ for various code parameters. We note that the uniform windowed decoder has the same threshold regardless of $\gamma_2$ and  $\mathcal{T}$. From the table, we see that $\epsilon^{worst}_{\delta,\mathcal{W}}$ and $\epsilon^\textnormal{WC}_{\delta,\mathcal{W}}$ are equal in the first four digits which indicates that $\epsilon^{worst}_{\delta,\mathcal{W}}$ is a good measure of the performance for our decoder. Therefore, for the same decoder complexity, our decoder is able to operate for $\epsilon$ for where the uniform windowed decoder would fail. Additionally, for the smallest $\gamma_2$ and $\mathcal{T}$, $\mathcal{W}$ has the highest threshold. Therefore, by exploiting the structure of an MD-SC code, the non-uniform windowed decoder is able to get closer to the performance of an optimal code for these simulations. We note that as $\gamma_2$ and $\mathcal{T}$ get smaller, the optimized window sizes tend to become larger around the TVNs. Intuitively, this is a result of the edges concentrating more closely to the sections near the TVNs. We can extrapolate that if $\gamma_2/L_2$ becomes very small, then many window sizes can be set to zero due to the large edge distance between the TVNs and the rest of the sections. Such an approach would not be possible with uniform windowed decoding.
	
	\subsection{Average Iteration Count}
	
	For this experiment, we show how the average number of iterations for all WCs, i.e., $\sum_{(i,j)\in[(L_1,L_2)]}\frac{I_{(i,j)}}{L_1L_2}$,   changes as a function of the window complexity, where $s(\mathcal{W})=C$. In Fig.~\ref{fig_average_iter}, we see the result of this experiment for various window complexities. Its worthwhile to mention that only $C=36$ and $C=45$ are achievable by a uniform windowed decoder and that all other points on the plot are only achievable by non-uniform windowed decoder. We observe that for small $\epsilon$, the average number of iterations is the same for all choices for $\mathcal{W}$. This is expected as $\epsilon$ is far from the BP thresholds of these WCs which is known to result in fast convergence. But as $\epsilon$ gets closer to the BP thresholds of each WC, the average number of iterations starts to split for the different complexities. We observe that for $\epsilon \approx 0.48$, the best WCs for $C=42$ and $C=45$ have about a $35\%$ reduction in average number of iterations compared to $C=36$. However, the non-uniform decoder achieves this improvement with only  an increase of $6$ sections in the window complexity where a uniform decoder requires at least an increase of $9$ sections.
	
	\begin{figure}[t]
		\centering
		\includegraphics[scale=0.54]{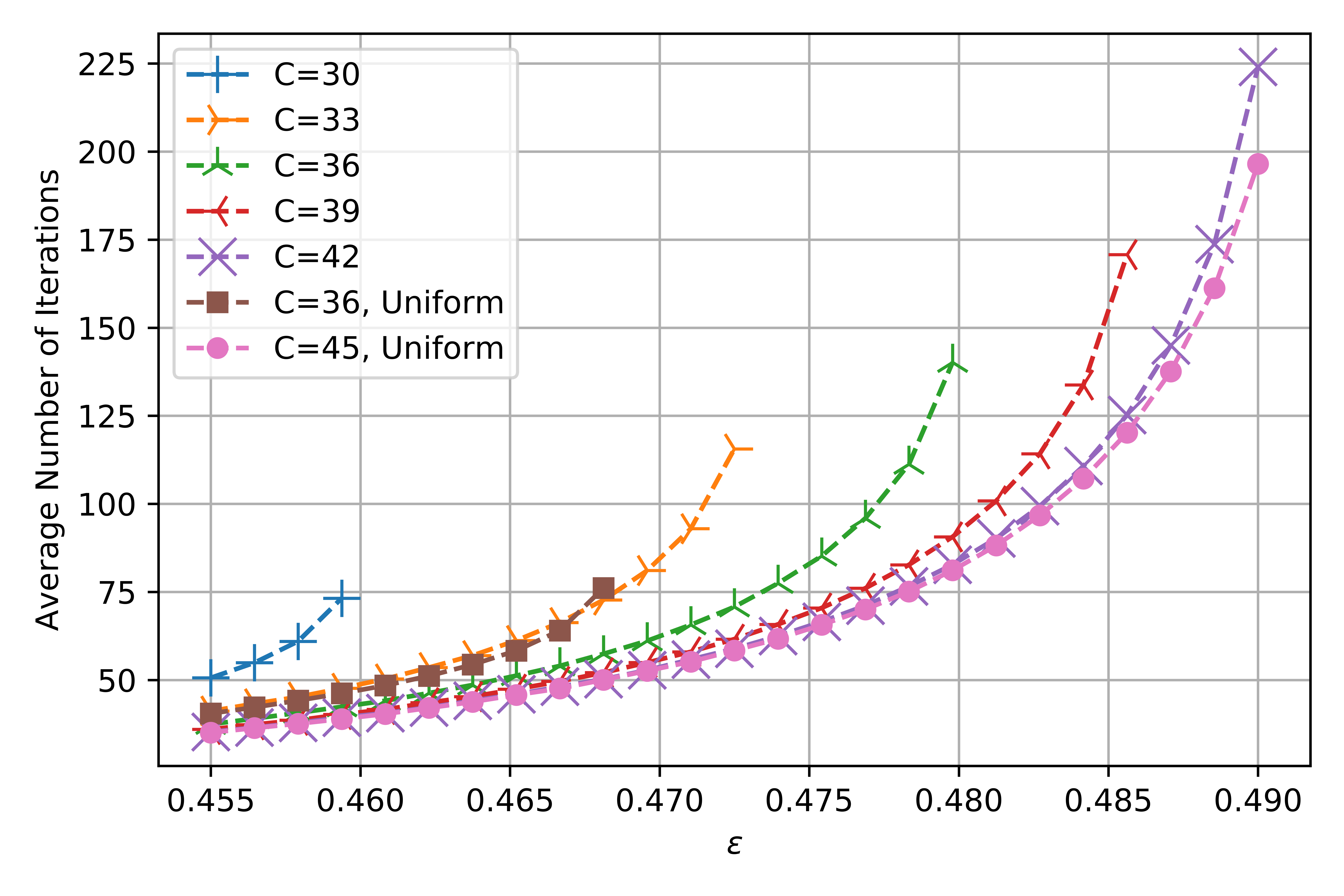}
		\vspace{-0.4cm}
		\caption{\footnotesize Average number of iterations across all WCs for various $\epsilon$. For each window complexity $C$, the best $\mathcal{W}$ was chosen based on worst-case BP thresholds with the window sizes constrained to $[2,5]$. The uniform WCs are presented for comparison. Each $\mathcal{W}$ was only evaluated for $\epsilon$ below its threshold. Code parameters are $d_l = 4$, $d_r = 8$, $\gamma_1=\gamma_2=2$, $L_1=30$, $L_2 = 9$, and $\delta= 10^{-12}$.\vspace{-0.3cm}}
		\label{fig_average_iter}
	\end{figure}
	
	\begin{figure}[t]
		\centering
		\includegraphics[scale=0.54]{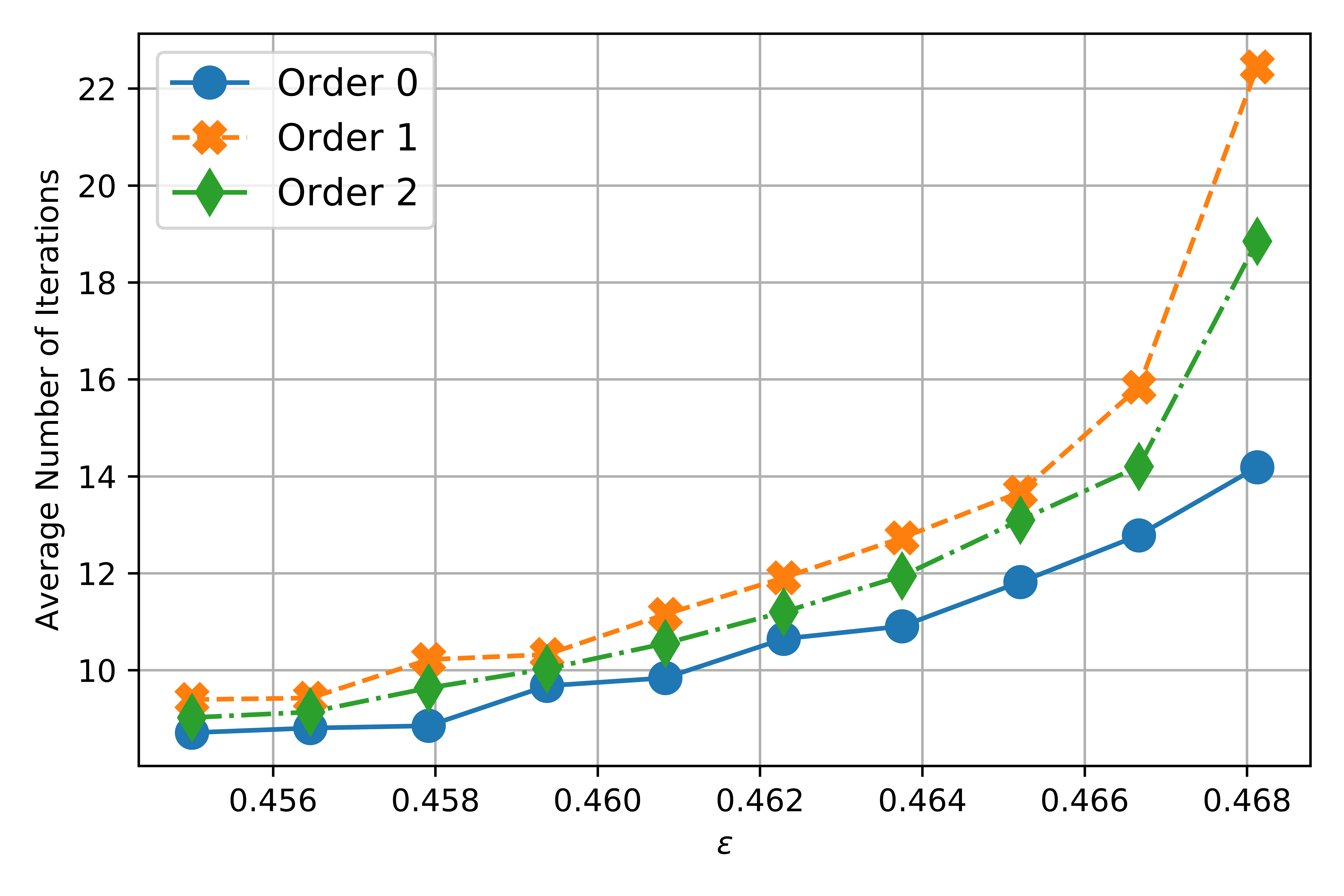}		\vspace{-0.4cm}
		\caption{\footnotesize Average number of iterations across all WCs for various $\epsilon$ and orderings.  Code parameters are $d_l = 4$, $d_r = 8$, $\gamma_1=\gamma_2=2$, $L_1=30$, $L_2 = 19$, $\delta= 10^{-12}$, and $\mathcal{W} = [5, 4, 3, 0,0,0,0,0,0,0,0,0,0,0,0,0,0, 3, 4]$.\vspace{-0.3cm}}
		\label{fig_average_order}
	\end{figure}
	
	\subsection{Effect of Processing Order}
	
	In this section, we quickly remark how the choice of the processing order for WC effects the average number of iterations across all WCs. In Fig.~\ref{fig_average_order}, we compare the average number of iterations for different orderings. Order $0$ is the natural order defined earlier (0 to $L_2-1$), Order $1$ is the reverse natural order ($L_2-1$ to $0$), and Order $2$ is a randomly chosen order. We observe that Order $0$ has the best performance among these orderings and that Order $1$ performs significantly worse than even a random ordering. \textcolor{black}{In this example, the window shape is symmetric so one would expect that going in either order $0$ or order $1$ would provide the same results. We postulate that this difference is partially due to the non-uniform coupling across the second dimension and other causes would require further study. However, the simulation supports our choice of natural ordering.}
	\section{Conclusion}
	
	In this paper, we defined a new variant of MD-SC-LDPC codes which offers more flexibility in designing windowed decoding. We proposed a novel windowed decoder using non-uniform window sizes which better exploits the structure of MD coupling. We demonstrated that, for certain cases, non-uniform windowed decoding can greatly improve the threshold while having the same complexity as uniform windowed decoding which allows our decoder to reliably operate at much higher channel erasure probabilities. Additionally, through simulations, we demonstrate how our decoder allows for a finer control over the latency and reliability trade-off.
	\section{Acknowledgments}
	
	Research supported in part by a grant from ASRC-IDEMA and grant CCF-BSF:CIF $\#1718389$ from NSF.
	\bibliographystyle{ieeetr}
	\bibliography{references}

\end{document}